\documentclass[10pt]{article}
\usepackage{graphicx} 
\usepackage{amsmath,amssymb,enumerate,xcolor}
\usepackage{amsthm}
\usepackage[hidelinks]{hyperref}

\title{Existence of Riemannian invariants for integrable systems of hydrodynamic type} 
\author{Alexey V.\ Bolsinov\footnote{School of Mathematics,
 Loughborough University,
 LE11 3TU, UK  and Institute of Mathematics and Mathematical Modeling, Almaty, Kazakhstan\ \ \quad {\tt A.Bolsinov@lboro.ac.uk}},    
Andrey Yu.\  Konyaev\footnote{Faculty of Mechanics and Mathematics and Center for Fundamental and Applied Mathematics, Moscow State University, 119992, Moscow, Russia 
 \ \ \quad{\tt  maodzund@yandex.ru}}  \, and 
    Vladimir S.\ Matveev\footnote{
Institut f\"ur Mathematik, Friedrich Schiller Universit\"at Jena,
07737 Jena,  Germany  \ \ \quad {\tt  vladimir.matveev@uni-jena.de}}}

\date{ }

\newtheorem{theorem}{Theorem}

\newtheorem{lemma}[theorem]{Lemma}
\newtheorem{Conjecture}{Conjecture}

\newcommand{\weg}[1]{}

\newcommand{\Id}{\operatorname{Id}}
\newcommand{\<}{\langle}
\renewcommand{\>}{\rangle}

\usepackage[
  backend=biber,
  style=numeric,
  sorting=nyt,
  isbn=false, doi=true, url=false  
]{biblatex}
\begin{document}

\maketitle
\begin{abstract}
We show that for a hyperbolic system of hydrodynamic type admitting $n$ symmetries, there exists a coordinate system in which the generator of the system and all the symmetries are diagonal.
\end{abstract}

\section{Introduction}
By an \emph{operator field} we understand a $(1,1)$-tensor field.
For a pair of commuting operator fields $L$ and $M$ (i.e., such that $LM = ML$), one can define a $(1,2)$-tensor field $\langle L, M \rangle$ by
\begin{equation}\label{eq:strongbracket}
\langle L, M\rangle(\xi, \eta) = M[L\xi, \eta] + L[\xi, M\eta] - [L\xi, M\eta] - LM[\xi, \eta].
\end{equation}

This definition is due to A.~Nijenhuis
\cite[formula 3.9]{nijenhuis1951}. Notice that $\langle L,M\rangle + \langle M,L\rangle$ coincides with the Fr\"olicher--Nijenhuis bracket of $L$ and $M$, and its output is a tensor field, also for operator fields that are not necessarily commuting.

We say that $M$ is a {\it symmetry} of an operator field $L$ if $LM = ML$ and the symmetric (in the lower indices) part of $\langle L, M \rangle$ vanishes, that is, $\langle L, M \rangle (\xi, \xi) = 0$ for all vector fields $\xi$. Clearly, if $L$ is a symmetry of $M$, then $M$ is a symmetry of $L$.
Our main result is the following theorem.

\begin{theorem}\label{thm:1}
Let $K_1, \dots, K_n$ be $n$-dimensional mutual symmetries. Assume that at a point $p$, they are linearly independent and there exists a linear combination $\sum c_iK_i$ having $n$ distinct real eigenvalues. Then, in a neighborhood of $p$ there exists a local coordinate system such that all $K_i$ are given by diagonal matrices.
\end{theorem}

This differential-geometric problem came from mathematical physics and is related to the integrability of PDE systems of hydrodynamic type; let us recall the terminology and the relation.
For an operator field $L$ in dimension $n$, consider the following quasilinear system of $n$ PDEs
for the unknown vector function $u=(u^1(x,t), \dots, u^n(x,t))^\top$:
\begin{equation} \label{eq:1}
u_t= L(u)u_x.
\end{equation}
The subscripts $x$ and $t$ denote partial derivatives, and the operator field $L$ is viewed as its matrix $L(u)$ in local coordinates $u^1,\dots, u^n$. Such systems are called \emph{of hydrodynamic type}; they model many processes in mathematical physics and appear in many pure mathematical problems, see e.g. \cite{Serre3,tsarev1}.
Note that it is natural to think of the matrix $L(u)$ as an operator field, since the transformation rule of the coefficient matrix
$L$ in \eqref{eq:1}, under a diffeomorphic change of the unknown functions, coincides with the transformation rule of $(1,1)$-tensor fields.

Then, the mutual symmetry condition reads as follows: consider the system of $2n$ equations
\begin{equation} \label{eq:2}
u_{t_1}= L(u)u_x\ , \ \ \ u_{t_2}= M(u)u_x,
\end{equation}
for the unknown vector function $u$, which now depends on three variables $(x,t_1,t_2)$:
$u=(u^1(x,t_1,t_2), \dots, u^n(x,t_1,t_2))^\top$. Then, $L$ and $M$ are mutual symmetries if and only if \eqref{eq:2} is formally compatible, i.e., has no differential obstructions. In the analytic category, this is equivalent to the existence and uniqueness of a local solution for any initial data $u(x,0,0)$. Respectively, the operators $K_1,\dots, K_n$ are mutual symmetries if the system
of $n^2$ equations $u_{t_\alpha}= K_\alpha(u)u_x$ on
$u=(u^1(x,t_1,\dots,t_n), \dots, u^n(x,t_1,\dots,t_n))^\top$ is formally compatible.

Local coordinates such that all $K_i$ are diagonal
are called \emph{Riemann invariants} in the literature. Their first appearance is indeed due to B.~Riemann (1860) \cite{Riemann},
who worked with systems of hydrodynamic type in dimension two, where (local) Riemann invariants always exist provided the operator field has two distinct eigenvalues at every point. 
Passing to coordinates in which $K$ is diagonal allowed him to simplify the problem and solve it using the so-called reciprocal transformation.

Systems of the form \eqref{eq:1} are well studied under the assumption of the existence of Riemann invariants. See \cite{Serre4,Serre3} for their study under the assumption of the existence of sufficiently many conservation laws (``rich systems'') and \cite{tsarev} under the assumption of the existence of nontrivial symmetries (``semi-Hamiltonian systems''), and also the survey \cite{tsarev1}. For hyperbolic systems admitting $n$ Riemann invariants, these two sets of assumptions, the existence of nontrivial symmetries and of sufficiently many conservation laws, are equivalent (in the real-analytic category) by \cite{sevennic,Sharipov}. Such systems are sometimes called Tsarev-integrable systems in the literature, as the so-called \emph{generalised hodograph method} \cite{tsarev} reduces solving them to a certain linear system of PDEs. It is known that Tsarev-integrable systems admit (locally) $n$ mutually commuting symmetries which are linearly independent at every point.
Our result shows that for systems whose corresponding operator has $n$ distinct real eigenvalues, the existence of Riemann invariants, which was an additional assumption in the above references, follows from the existence of $n$ mutual symmetries which are linearly independent at every point.

\section{Proof of Theorem \ref{thm:1}}
The main idea of the proof (used in, and possibly due to, \cite{BKM2025}) is as follows: the assumptions imply the existence of vector fields $v_1,\dots, v_n$ such that they are linearly independent at every point and such that they are eigenvector fields of all operators $K_i$. The existence of diagonal coordinates is equivalent to the condition that (for any $i,j$) the three vector fields
\begin{equation}\label{ew:3f}
v_i, v_j, [v_i, v_j]
\end{equation}
are linearly dependent. This condition is algebraic in the components of the vector fields and their first derivatives. Next, observe that the condition \eqref{eq:strongbracket} is also algebraic in the components of the operators and their first derivatives. We will show that the second condition implies the first. To do this at a fixed, arbitrarily chosen point, we may assume that the coordinate system is centered at this point and that the entries of $K_i$ are polynomials of degree one in the coordinates. Since replacing $K_i$ by their constant linear combinations does not affect the assumptions, we may assume that at the point where we work $K_i=E_{ii}$, that is, the only nonzero entry is the one in the $(i,i)$-position, and it equals $1$.

Next, observe that the condition \eqref{eq:strongbracket} is also algebraic in the components of the operators and their first derivatives, so for calculations this condition at $x=0$  we can again ignore higher order terms of entries of  $K_i$. We will show that  the condition \eqref{eq:strongbracket} evaluated at the point $x=0$ implies that \eqref{ew:3f} are linearly dependent at $x=0$. As the point can be chosen arbitrarily, this would prove Theorem \ref{thm:1}. 

We now give details of the calculations.
Let $x=(x_1,\dots,x_n)$ be a local coordinate system centered at zero (in the Introduction, in the part around \eqref{eq:1} and \eqref{eq:2}, the coordinates were denoted by $u^1,\dots, u^n$, since they represent the unknown functions in the system of PDEs). We denote by
\[
e_1=\partial_{x_1},\ \dots,\ e_n=\partial_{x_n}
\]
the corresponding coordinate vector fields. In dimensions $n=1,2$, Theorem \ref{thm:1} is trivial, so in what follows we assume $n\ge 3$.

For each $i\in\{1,\dots,n\}$ consider an $n\times n$ diagonal matrix field
\[
\widetilde K_i(x)= E_{ii} + D_i(x)
\]
where $D_i(x)$ is a diagonal matrix whose diagonal terms are linear in $x$ with no free terms.

Next, consider the matrix $A(x)$, whose entries are linear in $x$ with no free terms:
\[
A(x)=\sum_{m=1}^n x_m A^{(m)},\qquad A^{(m)}=\bigl(a^{(m)}_{pq}\bigr)_{p,q=1}^n\ \text{constant}.
\]
Define
\begin{equation}
\label{eq:b1}
K_i := (\Id-A)\widetilde K_i(\Id+A),\qquad i=1,\dots,n.
\end{equation}
Note that the matrices $K_i$ commute up to quadratic terms in $x$.
It is an easy exercise to show that for any commuting matrices $M_1(x),\dots, M_n(x)$ such that they coincide with $\widetilde K_1,\dots, \widetilde K_n$ at $0$, there exist coefficients
$\bigl(a^{(m)}_{pq}\bigr)_{p,q=1}^n$ such that for any $i$ the matrix $M_i$ coincides with $K_i$ up to terms of second and higher order in $x$. Moreover, there exists a system of common eigenvectors of the $M_i$ such that it coincides, up to second- and higher-order terms, with the system of
vectors $v_1,\dots, v_n$ given by
\begin{equation} \label{eq:vi}
v_i=(\Id - A)e_i.
\end{equation}
Thus, with an appropriate choice of $A$ and $D_1,\dots, D_n$,  formula \eqref{eq:b1} gives a linear approximation of our original operators $K_1,\dots, K_n$.

We need the following technical Lemma.

\begin{lemma}\label{lem:1}
Let
$$
\langle   K_i,  K_j\rangle_{|x=0}(\xi, \eta)+ \langle  K_i,  K_j\rangle_{|x=0}(\eta, \xi) =0, \quad i,j=1,\dots,n
$$
for any vectors $\eta, \xi$. Then, the coefficients of $a^{(m)}_{pq}$ satisfy
\begin{equation}\label{eq:B-system0}
 a^{(q)}_{rp}=a^{(p)}_{rq}\quad \text{for all pairwise distinct }r,p,q.
\end{equation}
\end{lemma}

\begin{proof}
By assumptions,   the vector 
\begin{equation} \< K_i, K_j\>_{|x=0}(e_i,e_j)+ \<  K_i, K_j\>_{|x=0}(e_j,e_i) \label{eq:comp}\end{equation} is zero.
Let us take mutually different $r,i,j$ and calculate the $r$th component of \eqref{eq:comp}. We assume $\widetilde K_i=E_{ii}+ D_i$ and $\widetilde K_j=E_{jj}+ D_j$, where $D_i$ and $D_j$ are diagonal matrices whose diagonal terms are linear in $x$.
We have:
$$
\begin{array}{l}\< K_i,  K_j\>_{|x=0}(e_i,e_j)\\ \hspace{1cm}  = \<(\Id- A) (E_{ii}+ D_i) (\Id+ A), (\Id- A) (E_{jj} + D_j) (\Id+ A)\>_{|x=0}(e_i,e_j).\end{array}
$$
Next, since the commutator of vector fields uses only one differentiation, the commutator of vectors whose entries are homogeneous in $x$ and whose degrees sum to $\ge 2$ vanishes after substituting $x=0$. Therefore, we obtain
\begin{equation}\label{eq:vanish}
\begin{array}{rl}\< K_i, K_j\>_{|x=0}(e_i,e_j)=&
\<E_{ii}, E_{jj}A-AE_{jj} + D_j\>(e_i,e_j)\\+&
\<E_{ii}A-AE_{ii} + D_i, E_{jj}\>(e_i,e_j).\end{array}
\end{equation}
Note that the right-hand side of \eqref{eq:vanish} is constant, so we do not need to substitute $x=0$ inside.
Next,
using $[e_i,e_j]=0$, we obtain from \eqref{eq:strongbracket}:
\begin{equation}\label{eq:concomitant-coords} \langle L,M\rangle(e_i,e_j)^r = L^b_i\,\partial_{x_b}M^r_j - M^b_j\,\partial_{x_b}L^r_i - L^r_c\,\partial_{x_i}M^c_j + M^r_c\,\partial_{x_j}L^c_i. \end{equation}
Plugging $L=E_{ii}$ and $M= E_{jj}A-AE_{jj} + D_j$ into this formula, we see that the last three summands on the right-hand side of \eqref{eq:concomitant-coords} vanish, so
the $r$th component of the term $\<E_{ii}, E_{jj}A-AE_{jj} + D_j\>(e_i,e_j)$ in \eqref{eq:vanish} reads
\begin{equation} \label{eq:term1}
\left(\<E_{ii}, E_{jj}A-AE_{jj} + D_j\>(e_i,e_j)\right)^r= -\partial_{x_i} A_{rj}  = - a^{(i)}_{rj}.
\end{equation}
Similarly, the $r$th component of the term $$\<E_{ii}A-AE_{ii} + D_i, E_{jj}\>_{|x=0}(e_i,e_j)$$ in \eqref{eq:vanish} reads $ a^{(j)}_{ri}$.
Next, consider the $r$th component of the second term $$ \<  K_i, K_j\>_{|x=0}(e_j,e_i)$$ 
of \eqref{eq:comp}. For this component, the analog of the formula \eqref{eq:concomitant-coords} reads
\begin{equation}\label{eq:concomitant-coords1} \langle L,M\rangle(e_j,e_i)^r = L^b{}_j\,\partial_{x_b}M^r{}_i - M^b{}_i\,\partial_{x_b}L^r_j - L^r_c\,\partial_{x_j}M^c_i + M^r_c\,\partial_{x_i}L^c_j. \end{equation}
Plugging $L=E_{ii}$ and $M= E_{jj}A-AE_{jj} + D_j$ into this formula, we see that the $r$th component of every summand in \eqref{eq:concomitant-coords1} vanishes.

Combining everything, we see that the $r$th component of the condition
$\langle  K_i, K_j\rangle_{|x=0}(e_i,e_j)+\langle  K_i, K_j\rangle_{|x=0}(e_j,e_i)=0$
implies
\eqref{eq:B-system0}.
\end{proof}

We are now able to prove Theorem \ref{thm:1}. Consider the vectors
$v_i$ given by \eqref{eq:vi}. As explained in the beginning, we need to check that $v_i(0)=e_i$, $v_j(0)=e_j$, and $[v_i,v_j]_{|x=0}$ are linearly dependent.
{In order to do this, let us now calculate $[v_i, v_j]$}.
Set $w_i(x):=A(x)e_i$. Then $v_i=e_i-w_i$ and $w_i$ is linear in $x$.
Hence
\begin{equation} \label{eq:4}
[v_i,v_j]=[e_i-w_i,e_j-w_j]=[e_j,w_i]+[w_j,e_i]+[w_i,w_j].
\end{equation}
Because the components $w_i,w_j$ are linear in $x$, their bracket is also linear and therefore vanishes at $x=0$.
The other two components on the right-hand side can be calculated in view of $(w_j)^p(x)=\sum_{m=1}^n a^{(m)}_{pj}x_m: $
\begin{equation} \label{eq:5} [e_i,w_j]=\partial_{x_i}w_j=\sum_{p=1}^n a^{(i)}_{pj}\,e_p, \ \ [w_i,e_j]=-\partial_{x_j}w_i=-\sum_{p=1}^n a^{(j)}_{pi}\,e_p. \end{equation}
Combining Lemma \ref{lem:1} and \eqref{eq:5}, we see that
all $a^{(i)}_{pj}e_p$ with mutually different $i,j,p$ pairwise 
cancel each other, implying
$
[v_i, v_j]_{|x=0}
\in \mathrm{span}\{e_i,e_j\}=\mathrm{span}\{v_i(0),v_j(0)\} .$
Theorem \ref{thm:1} is proved.

\subsection{ Discussion of the complex eigenvalues and of  Jordan blocks}

A natural analogue of Theorem \ref{thm:1} is also valid if some eigenvalues are complex conjugate (but there still exists a basis, involving complex-conjugate vectors, in which all $K_i$ are diagonal). Of course, in this case the coordinates corresponding to complex eigenvalues are complex-valued. Indeed, the proof of Theorem \ref{thm:1} is essentially algebraic and works over the field of complex numbers with no change. Let us now discuss the more complicated Jordan block case.

By direct computer calculations, using the same circle of ideas as in the proof of Theorem \ref{thm:1}, one obtains the following statement in dimensions $n=3,\dots,10$:

\emph{Let an operator field $K$ be conjugate, at every point, to the nilpotent Jordan $n\times n$ block, and let $K_1,\dots, K_n$ be polynomials in $K$ with coefficients depending on $x$, such that $K_1,\dots, K_n$ are linearly independent at every point. Assume that $\langle K_i, K_j\rangle(\xi,\xi)=0$ for any $\xi$, $i,j=1,\dots,n$. Then,
the Haantjes torsion of $K$, and therefore of all $K_i$, vanishes.}

Moreover, in dimensions $n=3,\dots,7$, direct calculations have shown a similar statement for operator fields $K$ whose Jordan form has precisely two Jordan blocks with different eigenvalues.

The computer algebra algorithm goes as follows. Arguing as in the beginning of the proof of Theorem \ref{thm:1}, it is sufficient to check this statement at $x=0$, assuming that $K=(\Id -A)\, J \,(\Id + A)$, where $J$ is the standard nilpotent Jordan block, the entries of $A$ are linear in $x$, and the coefficients of the polynomials $P_i(t)$ such that $P_i(K)=K_i$ are polynomials of first order in $x$ and satisfy $P_i(t)\big|_{x=0}=t^{i-1}$.

Then, viewing the coefficients of the entries of $A$ and the coefficients of the linear terms of the coefficients of the polynomials $P_i$ as unknowns, the condition that the $K_i$ are mutual symmetries at $x=0$ is a linear system of equations in these unknowns. The condition that the Haantjes torsion of $K$ vanishes at $x=0$ is also a linear system of equations in a subset of these unknowns. Using computer algebra, we checked that adding the second system of equations to the first one does not change the rank of the system. This implies that the second system follows from the first one, which proves the statement.

This suggests the following natural conjecture:
\begin{Conjecture}
Let $K_1, \dots, K_n$ be $n$-dimensional mutual symmetries that are linearly independent at every point.  Assume that there exists a linear combination $\sum c_iK_i$ which is $\mathrm{gl}$-regular. Then the Haantjes torsion of all $K_i$ vanishes.
\end{Conjecture}

\subsection*{Acknowledgments.}  A.\,B. was  supported by the Ministry of Science and Higher Education of the Republic of Kazakhstan (grant No. AP23483476), and   V.\,M. by  the DFG (project  529233771) and the ARC Discovery Programme DP210100951. We thank E. Ferapontov, P. Lorenzoni and  D. Serre for their comments. 

  Data sharing is not applicable to this article. The authors declare no conflicts of interest.

\printbibliography

\end{document}